\newcommand\longversion[1]{#1}
\newcommand\shortversion[1]{}

\shortversion{\documentclass{aamas2016}}

\longversion{
\documentclass[letterpaper]{article}
\usepackage{fullpage}
}

\usepackage[usenames,svgnames,dvipsnames]{xcolor}
\usepackage{graphicx}
\usepackage{tikz}
\usetikzlibrary{decorations.shapes,calc,matrix,decorations.markings,decorations.pathreplacing,patterns}

\usepackage{xspace}
\usepackage{multicol}

\usepackage{amsmath,amssymb}
\longversion{\usepackage{amsthm}}
\longversion{
}
\shortversion{\usepackage{times}}

\usepackage{enumitem}

\longversion{\usepackage{authblk}}

\newcommand{\argmin}{\operatornamewithlimits{argmin}}

\longversion{
\setlength{\parskip}{1.3ex plus 0.2ex minus 0.2ex}
\setlength{\parindent}{0pt}
}

\usepackage{multirow}
\usepackage{thmtools, thm-restate}
\usepackage{nicefrac}

\newcommand{\nfrac}{\nicefrac}

\newcommand{\FPT}{\ensuremath{\mathsf{FPT}}\xspace}
\newcommand{\WOH}{\ensuremath{\mathsf{W}[1]}-hard\xspace}
\newcommand{\WO}{\ensuremath{\mathsf{W}[1]}\xspace}

\newcommand{\NP}{\ensuremath{\mathsf{NP}}\xspace}
\newcommand{\NPH}{\ensuremath{\mathsf{NP}}-hard\xspace}
\newcommand{\NPC}{\ensuremath{\mathsf{NP}}-complete\xspace}

\newcommand{\Pshort}{\ensuremath{\mathsf{P}}\xspace}
\newcommand{\NPshort}{\ensuremath{\mathsf{NP}}\xspace}

\newcommand{\W}{\ensuremath{\mathsf{W}}\xspace}
\newcommand{\PARANPC}{para-\ensuremath{\mathsf{NP}}-complete\xspace}
\newcommand{\PARANPH}{para-\ensuremath{\mathsf{NP}}-hard\xspace}
\newcommand{\PARANP}{para-\ensuremath{\mathsf{NP}}\xspace}
\newcommand{\XP}{\ensuremath{\mathsf{XP}}\xspace}

\newcommand{\defparproblem}[4]{
  \vspace{3mm}
\noindent\fbox{
  \begin{minipage}{.95\linewidth}
  \begin{tabular*}{\linewidth}{@{\extracolsep{\fill}}lr} \textsc{#1}  & {\bf{Parameter:}} #3 \\ \end{tabular*}
  {\bf{Input:}} #2  \\
  {\bf{Question:}} #4
  \end{minipage}
  }
  \vspace{2mm}
}

\newcommand{\defproblem}[3]{
  \vspace{3mm}
\noindent\fbox{
  \begin{minipage}{.95\linewidth}
  \begin{tabular*}{\linewidth}{@{\extracolsep{\fill}}lr} \textsc{#1}  \\ \end{tabular*}
  {\bf{Input:}} #2  \\
  {\bf{Question:}} #3
  \end{minipage}
  }
  \vspace{2mm}
  }

\newcommand{\CC}{{\mathcal C}}

\newcommand{\GG}{{\mathcal G}}

\newcommand{\LL}{{\mathcal L}}

\newcommand{\QQ}{{\mathcal Q}}

\newcommand{\SSS}{{\mathcal S}}

\newcommand{\UU}{{\mathcal U}}
\newcommand{\VV}{{\mathcal V}}
\newcommand{\WW}{{\mathcal W}}
\newcommand{\XX}{{\mathcal X}}
\newcommand{\YY}{{\mathcal Y}}

\newcommand{\Disapprov}{Disapproval}
\newcommand{\Netdisapprov}{Net-Disapproval}

\newcommand{\disapprov}{disapproval}
\newcommand{\netdisapprov}{net disapproval}

\newtheorem{theorem}{Theorem}
\newtheorem{lemma}{Lemma}

\newtheorem{corollary}{Corollary}
\newtheorem{definition}{Definition}

\newtheorem{remark}{Remark}

\usepackage{bbm}
\newcommand{\eps}{\varepsilon}
\renewcommand{\epsilon}{\eps}

\usepackage[font=normalfont]{caption}

\shortversion{

\usepackage{setspace}
\setlength{\parindent}{0pt}
\setlength{\parskip}{0pt}

\newcommand{\subparagraph}{}
\usepackage{titlesec}

\titlespacing\section{0pt}{2pt plus 4pt minus 2pt}{2pt plus 2pt minus 2pt}
\titlespacing\subsection{0pt}{2pt plus 4pt minus 2pt}{2pt plus 2pt minus 2pt}
\titlespacing\subsubsection{0pt}{2pt plus 4pt minus 2pt}{2pt plus 2pt minus 2pt}

\usepackage{etoolbox}
\newcommand{\zerodisplayskips}{%
  \setlength{\abovedisplayskip}{3pt}
  \setlength{\belowdisplayskip}{3pt}
  \setlength{\abovedisplayshortskip}{3pt}
  \setlength{\belowdisplayshortskip}{3pt}}
\appto{\normalsize}{\zerodisplayskips}
\appto{\small}{\zerodisplayskips}
\appto{\footnotesize}{\zerodisplayskips}
\setlength{\textfloatsep}{3ex}
\setlength{\intextsep}{1ex}

\usepackage{subfig}
\captionsetup{belowskip=0ex,aboveskip=1ex}

\renewenvironment{proof}{\vspace{-\topsep}\noindent{\em Proof:}}{ \hfill $\square$\\ }
}

\shortversion{\title{On Choosing Committees Based on Approval Votes in the Presence of Outliers}}
\longversion{\newcommand{\papertitle}{On Choosing Committees Based on Approval Votes in the Presence of Outliers}}

\shortversion{
\pdfpagewidth=8.5truein
\pdfpageheight=11truein

\sloppy
}
\longversion{
\title{\papertitle}
\author{\shortversion{\vspace{-15pt}}Palash Dey, Neeldhara Misra, and Y. Narahari}
\affil{\shortversion{\vspace{-10pt}}Indian Institute of Science, Bangalore}
\affil{\shortversion{\vspace{-10pt}}\textit {\{palash,neeldhara,hari\}@csa.iisc.ernet.in}}
}

\begin{document}

\shortversion{
\numberofauthors{1}

\author{
\alignauthor
Paper id: 102
}
}

\maketitle

\shortversion{\vspace{-50pt}}

\begin{abstract}

We study the computational complexity of committee selection problem for several approval-based voting rules in the presence of outliers. Our first result shows that outlier consideration makes committee selection problem intractable for approval, net approval, and minisum approval voting rules. We then study parameterized complexity of this problem with five natural parameters, namely the target score, the size of the committee (and its dual parameter, the number of candidates outside the committee), the number of outliers (and its dual parameter, the number of non-outliers). For net approval and minisum approval voting rules, we provide a dichotomous result, resolving the parameterized complexity of this problem for all subsets of five natural parameters considered (by showing either \FPT or \WO{}-hardness for all subsets of parameters). For the approval voting rule, we resolve the parameterized complexity of this problem for all subsets of parameters except one.

We also study approximation algorithms for this problem. We show that there does not exist any $\alpha(\cdot)$ factor approximation algorithm for approval and net approval voting rules, for any computable function $\alpha(\cdot)$, unless $\Pshort=\NPshort$. For the minisum voting rule, we provide a $(1+\eps)$ factor approximation algorithm running in time $n^{O(\nfrac{\log m}{\eps^2})}$, for every constant $\eps>0$, where $m$ and $n$ are the number of candidates and the number of votes respectively.

\end{abstract}

\shortversion{
% Note that the category section should be completed after reference to the ACM Computing Classification Scheme available at
% http://www.acm.org/about/class/1998/.

\category{\\I.2.11}{Artificial Intelligence}{Distributed Artificial Intelligence}[Multiagent Systems]

%A category including the fourth, optional field follows...
%\category{D.2.8}{Software Engineering}{Metrics}[complexity measures, performance measures]

%General terms should be selected from the following 16 terms: Algorithms, Management, Measurement, Documentation, Performance, Design, Economics, Reliability, Experimentation, Security, Human Factors, Standardization, Languages, Theory, Legal Aspects, Verification.

\terms{Algorithms, Theory}

%Keywords are your own choice of terms you would like the paper to be indexed by.

}

\section{Introduction}

\sloppypar
Aggregating preferences of agents is a fundamental problem in artificial intelligence and social choice theory~\cite{conitzer2010makingdecisions}. Typically, agents (or voters) express their preferences over alternatives (or candidates). There are many different models for the expression of preferences, ranging from the simplistic (each voter provides his or her favorite choice, also known as \textit{plurality voting}) to the comprehensive (each voter provides a complete ranking over the set of all candidates). \textit{Approval ballots} are an intermediate model, where a voter approves or disapproves of each candidate --- thus a vote may be thought of as a subset of approved candidates or as a binary string indexed by the candidate set. Approval votes are considered as a good compromise between the extreme models --- they provide the agent an opportunity to make a comment about every candidate, without incurring the overhead of determining a full ranking on the candidate set~\cite{RePEc,Brams08,kilgour2010approval,handbook,baumeister2010computational}. 

Our work in this paper focuses on finding the ``best'' subset of $k$ candidates when given $n$ approval votes over $m$ candidates. The set of candidates together with all the votes is usually called an \emph{election instance} or simply an election. We use the term \emph{electorate} to refer to the set of votes of all the agents, and the term \emph{committee} to refer to a subset of $k$ candidates. Since a committee is a subset of candidates, and a vote can also be interpreted as a subset of the candidate set\footnote{This is simply the subset of approved candidates.}, one might consider various natural notions of distance between these two sets. A fixed notion of distance leads to a measure of suitability of a committee with respect to an election --- for instance, by considering the sum of distances to all voters, or the maximum distance incurred from any voter. For a committee $\XX$ and a vote $\SSS$, the notions of distance that are well-studied in the literature include the size of the symmetric difference (leading to the \textit{minisum} or \textit{minimax} rules)~\cite{brams2007minimax,legrand2007some,CaragiannisKM10,gramm2003fixed}, the size of the intersection (leading to the notion of \textit{approval score}), or the difference between $|\XX \cap \SSS|$ and $|\XX \setminus \SSS|$ (leading to the notion of \textit{net approval} scores)~\cite{procaccia2008complexity,SkowronF15,AzizBCEFW15}. 

\paragraph{Motivation} ~~The standard approach to the selection of a winning committee is to look for one that optimizes these scoring functions over the entire election --- that is, when the sum or max of scores are taken over all the votes. However, it is plausible that for some elections, there is a committee that represents a very good consensus when the scores are taken over a \textit{subset} of the voters rather than the entire set. For example, consider the approval voting rule mentioned above. The approval score of a committee $\XX$ is the sum of $|\XX \cap \SSS|$, taken over all votes $\SSS$ in the electorate. Let the \textit{approval score per vote} be the ratio of the approval score of $\XX$ to $n$. For a committee of size $k$, the best approval score per vote possible is $k$.  Clearly, the higher the approval score per vote, the greater a voter is `satisfied' on average. Now, consider the following toy example. Let $\LL \subseteq [m]$ be an arbitrary but fixed subset of $k$ candidates. Suppose the election has $(n-1)$ votes that are all exactly equal to $\LL$, and one vote that approves $[m]\setminus \LL$. The best approval score per vote that one can hope for here is $k(1 - 1/n)$, whereas if we restrict our attention to the election on the first $(n-1)$ votes alone, then $\LL$'s approval score per vote is $k$, which is now the best possible. Of course, the difference here is not significant for large $n$, but the illustration does provoke the following thought: is there a subset of at least $n^*$ votes that admit a committee whose approval score per vote is at least $k^*$? For $n^* = n$, we are back to the original question of finding the best committee. This more general setting, however, allows us to explore trade-offs and structure in the election: if there is a committee that prescribes a very good consensus over a large fraction of the election, then it is likely to be a more suitable choice for the community compared to a committee that optimizes the same score function over the entire election. While these committees may coincide (as in the toy example), there are easy examples where they would not, making this a question worth exploring.

The notion of finding a good consensus over a large subset has been explored in various contexts. In particular, the idea of using \textit{outliers} is quite common in the the literature of \textit{Closest String} problems. The setting of closest string involves a collection of $n$ strings, and the goal is to find a single string that minimizes either the maximum distance, or the sum of distances, from all the input strings. The most commonly studied notion of distance is the Hamming distance. Notice that once we interpret votes as binary strings, we are asking a very similar question, and the main distinction is that our search is only over strings that have a fixed number of ones. This similarity has been noted and explored in some works on voting before (see, for instance,~\cite{BS14}). In the context of Closest String, a question that is often asked is the following: given a budget $k$, is there a string that is ``close to'' at least $k$ strings? This question has been studied for both the minimax and minisum notions of closeness~\cite{BLLPW13,LKLB14}. Typically, the strings that are left out are called \textit{outliers}. In the context of voting, one imagines that there may be a few votes that express rather tangential opinions, and that a good consensus emerges once they are removed. We will also refer to such votes as outliers. In the context of social choice theory, the very notion of Young voting rule can be regarded as to finding minimum number of outliers whose removal makes some candidate Condorcet winner~\cite{young1978consistent,bartholdi1989voting,rothe2003exact,betzler2010parameterized}.

\paragraph{Our Framework} ~~One of the advantages of using scoring rules for approval ballots described above is that the winning committees are polynomially computable for most of them (with the notable exception of the minimax voting rule). However, once we pose the question of whether a target score $t$ is achievable after the removal of at most $k$ outliers, the complexity landscape changes dramatically. We show that this question is a computationally hard problem --- in particular, we establish \NP{}-hardness. Having shown hardness in the classical setting, we explore the complexity further, primarily from the perspective of exact exponential algorithms and approximation algorithms. In the context of the former, we use the framework of parameterized complexity. Briefly, a parameterized problem instance comprises of an instance~$x$ in the usual sense, and a parameter~$k$.  A problem with parameter~$k$ is called \emph{fixed parameter tractable} (\FPT{}) if it is solvable in time~$f(k) \cdot p(|x|)$, where~$f$ is an arbitrary function of~$k$ and~$p$ is a polynomial in the input size~$|x|$. While there have been important examples of traditional algorithms that have been analyzed in this fashion, the theoretical foundations for deliberate design of such algorithms, and a complementary complexity-theoretic framework for hardness, were developed in the late nineties~\cite{DF88,DF95a,ADF95,DF95b}. Just as \NP{}-hardness is used as evidence that a problem probably is not polynomial time solvable, there exists a hierarchy of complexity classes above \FPT{}, and showing that a parameterized problem is hard for one of these classes is considered evidence that the problem is unlikely to be fixed-parameter tractable. Indeed, assuming the Exponential Time Hypothesis, a problem hard for~\WO{} does not belong to \FPT{}~\cite{DowneyF99}.  The main classes in this hierarchy are:
$$\FPT{}  \subseteq \W{}[1] \subseteq \W{}[2] \subseteq \cdots \subseteq \W{}[P] \subseteq \XP,$$

where a parameterized problem belongs to the class~$XP$ if there exists an algorithm for it with running time bounded by~$n^{g(k)}$. A parameterized problem is said to be \PARANPC{} if it is \NPC{} even for constant values of the parameter. A classic example of a \PARANPC{} problem is graph coloring parameterized by the number of colors --- recall that it is \NPC{} to determine if a graph can be properly colored with three colors. Observe that a \PARANPC{} problem does not belong to~$XP$ unless~$\Pshort{} = \NP{}$. We are now ready to describe our contributions in greater detail. 

\subsection{Our Contributions}
% \paragraph{Our Contributions.} ~~
We consider three standard approval scoring mechanisms, namely the \textit{minisum, approval, and net approval scores}. The last two scores, as originally defined in the literature, are designed to simply give us the total amount of approval that a committee incurs from all the voters. The scores themselves are therefore non-decreasing functions of $n$, and as such, the question of outliers is not interesting if we use the scores directly (in particular, it is impossible to improve these scores by removing votes). Therefore, we consider the dual scoring system that complements the original --- namely, we score a committee based on the amount \textit{disapproval} that it incurs from all the votes, and seek to minimize the \textit{total disapproval}. Typically, for any notion of approval, there are either one or two natural complementary notions of disapproval that present themselves (discussed in greater detail below). This formulation is consistent with the idea that we want our scores to capture ``distance'' rather than closeness. We note that in terms of scores, these rules are equivalent to the original, but choosing to ask the minimization question allows us to formulate the problem of finding the best committee in the presence of outliers.  

\begin{remark}
One might also consider the approval score \textit{per vote} instead of the total score, as described earlier. However, we chose to use the notion of disapproval because of its consistency with the other distance-based rules (like minisum and minimax). All the variations are equivalent as scoring functions, and we note that our choice is only a matter of exposition.
\end{remark}

% \usepackage{multirow}
\begin{table}[t]
\shortversion{\renewcommand{\arraystretch}{1.8}}
\centering
\resizebox{\linewidth}{!}{
\begin{tabular}{|c|c|c|c|}
\hline
\multicolumn{2}{|c|}{{\bf{\shortversion{\Large} Measures of approval}}}                                             & \multicolumn{2}{c|}{{\bf{\shortversion{\Large} Measures of disapproval}}}                                                              \\ \hline
{\shortversion{\Large} Minisum }                      & \multicolumn{2}{c|}{{\shortversion{\Large}$|\XX \Delta \SSS|$}                                                                     } & {\shortversion{\Large} Minisum }                         \\ \hline
\multirow{2}{*}{{\shortversion{\Large}Approval}}              & \multirow{2}{*}{{\shortversion{\Large}$|\XX \cap \SSS|$}}                  & {\shortversion{\Large}$|\XX \setminus \SSS|$}                                       & \multirow{2}{*}{{\shortversion{\Large}Disapproval}}              \\ \cline{3-3}
                                       &                                                & {\shortversion{\Large}$|\SSS \setminus \XX|$}                                       &                                           \\ \hline
{\shortversion{\Large}Net Approval}                           & {\shortversion{\Large}$|\XX \cap \SSS| - |\XX \setminus \SSS|$}                 & {\shortversion{\Large}$|\XX \setminus \SSS| - |\XX \cap \SSS|$}                          & {\shortversion{\Large}Net Disapproval                          } \\ \hline
\end{tabular}
}
\caption{ The approval-based rules that are considered in this paper. Here, $\XX$ is a candidate committee and $\SSS$ is a vote, and the table illustrates the possible scores that $\XX$ can incur from $\SSS$. The total score of $\XX$ will be the sum of these scores taken over $\SSS$. }
\label{table:approvalscores}
\end{table}

%\multirow{2}{*}{Proportional Approval} & \multirow{2}{*}{$\sum_{i=1}^{|\XX \cap \SSS|} 1/i$} & $\sum_{i=1}^{|\XX \setminus \SSS|} 1/i$                      & \multirow{2}{*}{Proportional Disapproval} \\ \cline{3-3}
%                                       &                                                &$\sum_{i=1}^{|\SSS \setminus \XX|} 1/i$ &                                           \\ \hline

In~Table\nobreakspace \ref {table:approvalscores}, we summarize the notions of distances between a committee and a vote. Each of these notions naturally gives rise to a score-based voting rule. Formally, for any distance function $s: 2^{\CC} \times 2^{\CC} \rightarrow \mathbb{N}$ between two subsets of candidates, we overload notation and define the corresponding score function $s: 2^{\CC} \times \VV \rightarrow \mathbb{N}$ for a set of votes $\VV = \{\SSS_1, \ldots, \SSS_n\}$ as follows:
\[ s(\XX,\VV) := \sum_{i=1}^n s(\XX,\SSS_i) \]
For the \textit{winner determination} problem, the goal is to find a committee $\XX$ of size $m^*$ that minimizes $s(\XX,\VV)$. For all the scoring rules in~Table\nobreakspace \ref {table:approvalscores}, a winning committee of size $m^*$ can be found in polynomial time for any $m^* \leq m$. We refer the reader to Table\nobreakspace \ref {tbl:notations} for an overview of the notation we use in this paper. We are now ready to define the problem of winner determination for a scoring rule $s$ in the presence of outliers:

\defproblem{$s$-\textsc{Outliers}}{A set of votes $\mathcal{V}=\{\mathcal{S}_1, \dots, \mathcal{S}_n\}$ over a set of candidates $\mathcal{C}=\{c_1, \dots, c_m\}$, a committee size $m^*$, a target number of non-outliers $n^*$, and a target score $t$.}{Does there exist a committee $\mathcal{C}^*\subset \mathcal{C}$ and a set of non-outliers $\mathcal{V}^*\subset \mathcal{V}$ such that $|\mathcal{V}^*|\ge n^*$, $|\mathcal{C}^*|=m^*$, and $s(\mathcal{V}^*, \mathcal{C}^*) \le t$?}

\begin{remark}
We will focus on only one variant of disapproval for the approval scoring rule, namely the one given by the top row in~Table\nobreakspace \ref {table:approvalscores}. The other variation is symmetric and our results hold in the exact same fashion, as explained in Lemma~\ref{lem:translate-disapprov} (in Section~\ref{sec:conclusions}).	
\end{remark}

We first show that the $s$-\textsc{Outliers} problem is \NPC{} for all the scoring rules considered here, even in the special case when every vote approves exactly three candidates \textit{and} every candidate is approved by exactly three votes. This also establishes the \PARANP{}-hardness of the problem with respect to the (combined) parameters \textit{(maximum) number of candidates approved by any vote} and \textit{(maximum) number of votes that approve a candidate}. 

\begin{restatable}{theorem}{NPCtheorem}
\label{thm:npc}
Let $s \in \{$Minisum, \Disapprov{}, \Netdisapprov{}$\}$. The $s$-\textsc{Outliers} problem is \NPC even if every vote approves exactly $3$ candidates and every candidate is approved by exactly $3$ votes.
\end{restatable}

To initiate the parameterized study of the $s$-\textsc{Outliers} problem, we propose the following parameters: the size of the committee ($m^*$), the number of candidates \textit{not} in the committee ($\overline{m}$), the number of non-outliers ($n^*$), the number of outliers ($\overline{n}$), and the target score $t$ (which one might think of as the ``solution size'' or the \textit{standard parameter}).

\begin{table}[h]
\centering
  \begin{center}
    \resizebox{\linewidth}{!}{
	{\renewcommand{\arraystretch}{1.8}
	\begin{tabular}{|cc||cc|}\hline
	  {\shortversion{\Large} Number of votes}		& {\shortversion{\Large} $n$}	& {\shortversion{\Large} Set of candidates}		& {\shortversion{\Large} $\mathcal{C}$}				\\\hline
	  {\shortversion{\Large} Number of candidates} 	& {\shortversion{\Large} $m$} 		& {\shortversion{\Large} Committee chosen}		& {\shortversion{\Large} $\mathcal{C}^*$}			\\\hline 
	  {\shortversion{\Large} Number of outliers}	& {\shortversion{\Large} $\overline{n}$}& {\shortversion{\Large} Non-committee}	& {\shortversion{\Large} $\overline{\mathcal{C}}(=\mathcal{C}\setminus\mathcal{C}^*)$}\\\hline
	  {\shortversion{\Large} Number of non-outliers}& {\shortversion{\Large} $n^* (=n-\overline{n})$}& {\shortversion{\Large} Set of votes}		& {\shortversion{\Large} $\mathcal{V}$}	\\\hline 
	  {\shortversion{\Large} Size of committee}		& {\shortversion{\Large} $m^*$}	&{\shortversion{\Large} Set of non-outliers}	& {\shortversion{\Large} $\mathcal{V}^*$}	 		\\\hline
	  {\shortversion{\Large} Size of non-committee}& {\shortversion{\Large} $\overline{m} (=m-m^*)$}& {\shortversion{\Large} Set of outliers}		& {\shortversion{\Large} $\overline{\mathcal{V}} (=\mathcal{V} \setminus \mathcal{V}^*)$}	\\\hline
	  {\shortversion{\Large} Score of the committee} 	& {\shortversion{\Large} $t$}	& {\shortversion{\Large} Hamming distance} & {\shortversion{\Large} $h(\cdot)$}	\\\hline	
	\end{tabular}
	}
   }
    \caption{\normalfont Notation table.}\label{tbl:notations}
  \end{center}
\end{table} 

Our main results are the following. For any subset of these parameters, we establish if the $s$-\textsc{Outliers} problem is \FPT{} or \WOH{} when $s \in \{\mbox{Minisum, Net Approval}\}$, and for the Disapproval voting rule, we have a classification for all cases but one. Specifically, for the minisum voting rule, we establish the following dichotomy.

\begin{restatable}{theorem}{DichotomyTheorem}
\label{thm:main1}	
	Let ${\mathcal P} = \{m^*, \overline{m}, n^*, \overline{n}, t\}$, and ${\mathcal Q} \subseteq {\mathcal P}$. The \textsc{Minisum-outliers} problem parameterized by ${\mathcal Q}$ is \FPT{} if ${\mathcal Q}$ contains either $\{m^*,\overline{m}\}$, $\{n^*,\overline{n}\}$, or $\{\overline{n}, t\}$, and is \WOH{} otherwise.
\end{restatable}

For the net disapproval voting rule, we establish the following dichotomy.

\begin{restatable}{theorem}{DichotomyTheoremNetApproval}
\label{thm:main2}	
	Let ${\mathcal P} = \{m^*, \overline{m}, n^*, \overline{n}, t\}$, and ${\mathcal Q} \subseteq {\mathcal P}$. The \textsc{\Netdisapprov{}-outliers} problem parameterized by ${\mathcal Q}$ is \FPT{} if ${\mathcal Q}$ contains either $\{m^*,\overline{m}\}$ or $\{n^*,\overline{n}\}$, and is \WOH{} otherwise.
\end{restatable}

Further, for the disapproval voting rule, we have the following theorem that classifies all cases but one.

\begin{restatable}{theorem}{NearDichotomyTheorem}
\label{thm:main3}	
	Let  ${\mathcal P} = \{m^*, \overline{m}, n^*, \overline{n}, t\}$, and ${\mathcal Q} \subseteq {\mathcal P}$. The \textsc{\Disapprov{}-Outliers} problem parameterized by ${\mathcal Q}$ is \FPT{} if ${\mathcal Q}$ contains either $\{m^*,\overline{m}\}$ or $\{n^*,\overline{n}\}$, and is \WOH{} for all other cases except for ${\mathcal Q} = \{ \overline{m}, \overline{n}, t \}$.
\end{restatable}

Apart from these classification results in the context of exact algorithms, we also pursue approximation algorithms and other special cases. We briefly summarize our contributions in these contexts below. 

\textit{Approximation Results.} We provide a polynomial time $\eps \overline{m}$-approximation algorithm for optimizing the minisum score considering outliers, for every constant $\eps>0$ [Theorem\nobreakspace \ref {thm:Minisum_approx}]. We also show a $(1+\eps)$-approximation algorithm for optimizing the minisum score considering outliers running in time $n^{O(\nfrac{\log m}{\eps^2})}$, for every constant $\eps>0$ [Theorem\nobreakspace \ref {thm:Minisum_as}]. On the hardness side, we show that in the presence of outliers, there does not exist an $\alpha(m,n)$-aproximation algorithm for optimizing score of the selected committee for any computable function $\alpha(\cdot,\cdot)$, for the \netdisapprov{} [Corollary\nobreakspace \ref {thm:no_apx_net}] and \disapprov{} voting rules (unless $\Pshort=\NPshort$) [Lemma\nobreakspace \ref {thm:no_apx_other}].

\textit{Other Special Cases.} We show that when every voter approves at most one candidate, then the \textsc{Minisum-outlier} problem can be solved in polynomial amount of time. For some \WOH{} cases, we show that the problem becomes \FPT{} if the maximum number of candidates approved by a vote or the maximum number of votes that approve of a candidate is a constant. We refer the reader to~Section\nobreakspace \ref {sec:conclusions} for a more detailed discussion. 

\subsection{Related Work}
% \paragraph{Related Work.} ~~
The notion of outliers is quite prominent in the literature pertaining to closest strings, where the usual setting is that we are given $n$ strings over some alphabet $\Sigma$, and the task is to find a string $x$ that minimizes either the maximum Hamming distance from any string, or the sum of its Hamming distances from all strings. In the context of social choice theory, the notion of outliers are intimately related to manipulation and control of election by removing voters and candidates~\cite{bartholdi1992hard,bartholdi1989computational,conitzer2007elections,chevaleyre2007short,hemaspaandra2007anyone,brandt2012computational,faliszewski2009richer,faliszewski2011multimode,meir2008complexity,procaccia2007multi,elkind2009distance,betzler2008parameterized}.

In the minimax notion of distance, the problem of finding a closest string in the presence of outliers was initiated in~\cite{BM11}, and was further refined in~\cite{BLLPW13}. The work in~\cite{BLLPW13} demonstrates hardness of approximation for variations involving outliers. With the minisum notion of distance, the study of finding the best string with respect to outliers was studied in~\cite{LKLB14}. In contrast with the minimax rule, this work shows a PTAS for \textit{Consensus Sequence with Outliers}. 

The fact that the closest string problem has similarities to the setup of approval ballots has been explored in the past, see, for instance~\cite{BS14}. More specifically, and again in the spirit of observing similarities with the closest string family of problems, the minimax approval rule has been studied from the perspective of outliers~\cite{MNS15}. To the best of our knowledge, our work is the first comprehensive study of the complexity behavior of approval ballots in the presence of outliers.

\section{Preliminaries}\label{sec:prelim}

In this section we briefly recall some terminology and notation that we will use throughout. 

\paragraph{Approval Voting} ~~Let $\mathcal{V}=\{\mathcal{S}_1, \dots, \mathcal{S}_n\}$ be the set of all \emph{votes} and $\mathcal{C}=\{c_1, \dots, c_m\}$ the set of all \emph{candidates}. Each vote $\mathcal{S}_i$ is a subset of $\mathcal{C}$. We say voter $i$ approves a candidate $x\in \mathcal{C}$ if $x\in \mathcal{S}_i$; otherwise we say, voter $i$ does not approve the candidate $x$. A voting rule $r$ for choosing a committee of size $m^*$ is a mapping $r:\cup_{n\ge1}(2^\mathcal{C})^n\longrightarrow \mathbb{P}_{m^*}(\mathcal{C})$, where $\mathbb{P}_{m^*}(\mathcal{C})$ is the set of all subsets of $\mathcal{C}$ of size $m^*$. We call a voting rule $r$ a {\it scoring rule} if there exists a scoring function $s:\cup_{n\ge1}(2^\mathcal{C})^n \times 2^\mathcal{C} \longrightarrow \mathbb{N}$ such that $r(\mathcal{V}) = \argmin_{X\in \mathbb{P}_{m^*}(\mathcal{C})} s(\mathcal{V},X)$, for every $\mathcal{V}\in(2^\mathcal{C})^n$. 
We refer to Table\nobreakspace \ref {table:approvalscores} for some common scoring rules. Unless mentioned otherwise, we use the notations listed in Table\nobreakspace \ref {tbl:notations}.

\paragraph{Algorithmic Terminology} ~~Given a set $A$, we denote the complement of $A$ by $A^c$. For a positive integer $\ell,$ we denote the set $\{1, 2, \cdots, \ell\}$ by $[\ell]$. We use the notation $O^*(f(m,n))$ to denote $O(f(m,n)poly(m,n))$.

For a minimization problem $\mathcal{P}$, we say an algorithm $\mathcal{A}$ archives an approximation factor of $\alpha$ if $\mathcal{A}(\mathcal{I}) \le \alpha OPT(\mathcal{I})$ for every problem instance $\mathcal{I}$ of $\mathcal{P}$. In the above, $OPT(\mathcal{I})$ denotes the value of the optimal solution of the problem instance $\mathcal{I}$.

\begin{definition}{\rm \bf(Parameterized Reduction~\cite{downey1999parameterized})}
\label{def:ppt-reduction} \\
Let $\Gamma_1$ and $\Gamma_2$ be parameterized problems. A parameterized reduction from $\Gamma_1$ to $\Gamma_2$ is an algorithm that, given an instance $(x,k)$ of $\Gamma_1$, outputs an instance $(x^\prime, k^\prime)$ of $\Gamma_2$ such that:
\begin{enumerate}[topsep=0pt,itemsep=0pt]
 \item $(x,k)$ is a yes-instance of $\Gamma_1$ if and only if $(x^\prime, k^\prime)$ is a yes-instance of $\Gamma_2$,
 \item $k^\prime \le g(k)$ for some computable function $g$, and
 \item the running time of the algorithm is $f(k)|x|^{O(1)}$ for some computable function $f$.
\end{enumerate}
\end{definition}

\section{Classical Complexity: \NP{}-hardness Results}

We begin by showing that even for rules where winner determination is polynomial time solvable, the possibility of choosing outliers makes the winner determination problem significantly harder. In particular, we show that the $s$-\textsc{Outliers} problem is hard for minisum, disapproval, and net disapproval voting rules. We reduce from the well known \textsc{Vertex Cover} problem, which is known to be \NPC{} even on $3$-regular graphs~\cite{garey1979computers}.

%\defproblem{\textsc{Vertex-cover}}{A graph $G$ and a positive integer $k$.}{
%Is there a subset $S$ of at most $k$ vertices such that $G[V \setminus S]$ is an independent set? 
%}

\defproblem{$3$-\textsc{Vertex-cover}}{A $3$-regular graph $G$ and a positive integer $k$.}{Is there a subset $S$ of at most $k$ vertices such that $G[V \setminus S]$ is an independent set? 
}

\NPCtheorem*
\begin{proof}
 First let us prove the result for the minisum voting rule. The problem is clearly in \NP. To prove \NP-hardness, we reduce the $3$-\textsc{Vertex-cover} problem to the \textsc{Minisum-outliers}. Let $(\mathcal{G}=(U,E), k)$ be an arbitrary instance of $3$-\textsc{Vertex-cover} problem. Let $U = \{u_1, \cdots, u_n\}$. We will introduce a candidate and a vote corresponding to every vertex, and have candidate $i$ approved by a vote $j$ if and only if the vertices corresponding to $i$ and $j$ are adjacent in $G$. We define the corresponding instance $(\mathcal{V}, \mathcal{C}, n^*, m^*, t)$ of the \textsc{Minisum-outliers} problem as follows: 
 \[ \mathcal{V} = \{\mathcal{S}_1, \cdots, \mathcal{S}_n\}, \mathcal{C} = \{c_1, \cdots, c_n\}, \]
 \[  S_i = \{ c_j : (u_i, u_j)\in E \} ~\forall i\in [n], \]
 \[n^* = n-k, m^* = k, t = (n-k)(k-3) \] 
 Notice that, $|\mathcal{S}_i| = 3$ and each candidate is approved in exactly $3$ votes since $\mathcal{G}$ is $3$-regular. We claim that the two instances equivalent. In the forward direction, suppose $W\subset U$ forms a vertex cover with $|W|=k$. Consider the committee $\mathcal{C}^* = \{ c_i : u_i \in W \}$ and the set of nonoutliers $\mathcal{V}^* = \{ S_i : u_i \notin W \}$. We claim that $h(\mathcal{S}_i, \mathcal{C}^*) = (k-3)$ for every $\mathcal{S}_i \in \mathcal{V}^*$ thereby proving $h(\mathcal{V}^*, \mathcal{C}^*) = (n-k)(k-3).$ Notice that, $h(\mathcal{S}_i, \mathcal{C}^*) \ge (k-3)$ since $|\mathcal{C}^*| = k$ and $|\mathcal{S}_i| = 3$ for every $\mathcal{S}_i\in \mathcal{V}^*$. Suppose, there exist an $\mathcal{S}_i \in \mathcal{V}^*$ such that $h(\mathcal{S}_i, \mathcal{C}^*) < k-3.$ Then there exist a candidate $c_j\in \mathcal{S}_i \setminus \mathcal{C}^*.$ However, this implies that the edge $(u_i, u_j) \in E$ is not covered by $W$ which contradicts that fact that $W$ is a vertex cover.
 
 For the reverse direction, suppose there exists a set of outliers $\overline{\mathcal{V}} \subset \mathcal{V}$ and a committee $\mathcal{C}^* \subset \mathcal{C}$ such that $|\overline{\mathcal{V}}|\le \overline{n}$, $|\mathcal{C}^*|=m^*$, and $h(\mathcal{V}^*, \mathcal{C}^*) \le t$. Since, adding votes in the set of outliers $\overline{\mathcal{V}}$ can only reduce $h(\mathcal{V}^*, \mathcal{C}^*)$, we may assume without loss of generality that $|\overline{\mathcal{V}}| = \overline{n}.$ Now, since $h(\mathcal{S}_i, \mathcal{C}^*) \ge (k-3)$ for every $\mathcal{S}_i\in \mathcal{V}^*$, we have $h(\mathcal{V}^*, \mathcal{C}^*) = (n-k)(k-3)$ and thus $h(\mathcal{S}_i, \mathcal{C}^*) = (k-3)$ for every $\mathcal{S}_i\in \mathcal{V} \setminus \overline{\mathcal{V}}.$ Let $\XX = \{ u_i\in U : \mathcal{S}_i \in \mathcal{V}^* \}$ and $\YY = \{ u_i \in U : c_i \in \mathcal{C}^* \}$. 
 We claim that $\XX$ covers all the edges incident on the vertices in $\YY$. Indeed, otherwise there exist a $u_i \in \XX$ such that all the edges incident on $u_i$ are not covered by $\YY$ and thus we have $h(\mathcal{S}_i, \mathcal{C}^*) > k-3$. We now construct a vertex cover $\WW$ of $\GG$ of size at most $k$ as follows. If $\XX$ and $\YY$ are disjoint then, $\YY$ forms a vertex cover of $\GG$. Otherwise we let $\YY_0=\YY, u\in\XX\cap\YY_0, v\in(\UU\setminus\XX)\setminus\YY_0$ and notice that $\YY_0\setminus\{u\}$ still covers all the edges incident on $\XX$. We define $\YY_1 = (\YY_0 \setminus\{u\}) \cup\{v\}$. We note that $|\XX\cap\YY_1| = |\XX\cap\YY_0|-1$ and $|\YY_1|=|\YY_0|$. We iterate this process for $|\XX\cap\YY|$ many times. Let us call $\YY_{|\XX\cap\YY|}$ $\WW$. By the argument above, $\WW$ covers all the edges incident on the vertices in $\XX$, $\WW\cap\XX=\emptyset$, and thus we have $\WW = \UU\setminus\XX$. Hence, $\WW$ forms a set cover of $\GG$.
 
 The proofs for the other rules are identical, except for the values of the target score. We define $t=(n-k)(k-3)$ for the \disapprov{}, and $t=(n-k)(k-6)$ for the \netdisapprov{} voting rule. It is easily checked that the details are analogous.
 \end{proof}

\longversion{
\begin{figure}[h]
	\centering 

\begin{tikzpicture}[scale=0.6]
\draw[DarkSlateGrey,dashed] (-5.2,10.2) rectangle (10.2,-5.2);

\draw[DarkSlateGrey,fill=IndianRed!20] (0,0) rectangle (10,10);
\draw[DarkSlateGrey,fill=YellowGreen!20] (-0.2,0) rectangle (-5,10);

\draw[decorate, decoration={brace, raise=6pt}]
(0,10) -- (10,10);

\draw[decorate, decoration={brace, raise=6pt}]
(-5,10) -- (-0.2,10);

\draw[decorate, decoration={brace, mirror, raise=6pt}]
(-5,10) -- (-5,0);

\draw[decorate, decoration={brace, mirror, raise=6pt}]
(-5,-0.2) -- (-5,-5);

\draw[decorate, decoration={brace, raise=6pt}]
(0,10) -- (10,10);

\node at (5,11) {Outside the Committee};
\node at (-2.5,11) {Committee};
\node[rotate=90] at (-6,5) {Non-outliers};
\node[rotate=90] at (-6,-2.5) {Outliers};

\node at (2.5,-6) {Adjacency Matrix of $G$.};
\node at (5,5) {Independent Set. (All Zeroes)};

\draw[DarkSlateGrey,fill=DodgerBlue!20] (-5,-0.2) rectangle (10,-5);

\draw[<->]
(-5,12) -- (10,12);
\node at (2.5,12.5) {$V(G)$};

\draw[<->]
(-7,10) -- (-7,-5);
\node[rotate=90] at (-7.5,2.5) {$V(G)$};

%
%\node[rotate = 90,yshift=20pt] at (my-12-1)
%    {\textsc{Outliers}};
%
%  \draw[decorate, decoration={brace, mirror, raise=6pt}]
%    (my-1-1.north west) -- (my-11-1.north west);
%\node[rotate = 90,yshift=20pt] at (my-5-1)
%    {\textsc{Non Outliers~~}};

\end{tikzpicture}
\end{figure}
}

Theorem\nobreakspace \ref {thm:npc} immediately yields the following corollary.

\begin{corollary}
Let $s \in \{$Minisum, \Disapprov{}, \Netdisapprov{}$\}$. Then the $s$-\textsc{Outliers} problem is \PARANPH{} when parameterized by the sum of the maximum number of approvals that a candidate obtains and the maximum number of candidates that a vote approves.
\end{corollary}

\section{Parameterized Complexity Results}

In this section, we present the results on parameterized complexity of the $s$-\textsc{Outliers} problem.

\subsection{\FPT{} algorithms}

The following result follows from the fact that the $s$-\textsc{Outliers} problem is polynomial time solvable for all the voting rules considered here if we know either the committee or the non-outliers of the solution. \shortversion{We skip proof of some results in the interest of space}

\begin{restatable}{proposition}{FPTmnLemma}\label{prop:FPT_m_n}
\label{prop:fpt_mn}	
Let $s \in \{$Minisum, \Disapprov{}, \Netdisapprov{}$\}$. Then, there is a $O^*(2^m)$ time algorithm and a $O^*(2^n)$ time algorithm for the $s$-\textsc{Outliers} problem.\end{restatable}

Now, we show that the \textsc{Minisum-outliers} problem, parameterized by $(t + \overline{n})$, is \FPT{}.

\begin{lemma}\label{lem:fpt_t}
 There is a $O^*(2^{t+\overline{n}})$ time algorithm for the \textsc{Minisum-outliers} problem.
\end{lemma}

\begin{proof}
 We consider the following two cases.
 
 {\bf Case $t < n - \overline{n}$: } In this case there exist a vote say $\mathcal{S}_i$ such that $h(\mathcal{X}, \mathcal{S}_i)=0$, where $\mathcal{X}$ is the committee selected. 
 Hence, we can iterate over all possible such $\mathcal{S}_i$ (if $|\mathcal{S}_i| \ne m^*$ then, do not consider that $\mathcal{S}_i$) and fix $\mathcal{X}$ to be $\mathcal{S}_i$. We now identify the non-outliers $\mathcal{V}^*$ with respect to the committee $\mathcal{X}$.
 
 {\bf Case $t \ge n - \overline{n}$:} Use the algorithm in~Proposition\nobreakspace \ref {prop:fpt_mn} which runs in $O^*(2^n) = O^*(2^{t+\overline{n}})$ time.
\end{proof}

\subsection{\WO{}-hardness Results}

In this section, we establish our \WO{}-hardness results. To begin with, we focus on parameters combined with the target score $t$. Notice that we have tractability when we combine either $(\overline{n} + n^*)$ or $(\overline{m} + m^*)$ along with $t$ (follows from Proposition\nobreakspace \ref {prop:FPT_m_n}). For minisum, we even have tractability for $t$ and $\overline{n}$, from Lemma\nobreakspace \ref {lem:fpt_t}. Therefore, the interesting combinations for the minisum rule are $(t + n^* + \overline{m})$ and $(t + n^* + m^*)$. We first consider the combination $(t+n^*+\overline{m})$, and show \WO{}-hardness for all the voting rules considered here by exhibiting a parameterized reduction from the $d$-\textsc{Clique} problem which is known to be \WOH parameterized by clique size~\cite{Downey2013}.

\defparproblem{$d$-\textsc{Clique}}{A $d$-regular graph $G$ and a positive integer $k$.}{$k$}{Is there a clique of size $k$?}

\begin{lemma}\label{lem:hard_mbar_nstar}
The \textsc{Minisum-outliers} problem is \WOH{}, when parameterized by $(t + n^*+\overline{m})$. Also, for $s \in \{$\disapprov{}, \netdisapprov{}$\}$, the $s$-\textsc{Outliers} problem is \WOH{} when parameterized by $(n^*+\overline{m})$, even when $t=0$.
\end{lemma}

\begin{proof}
 First let us prove the result for the \textsc{Minisum-outliers} problem. We exhibit a parameterized reduction from the $d$-\textsc{Clique} problem to the \textsc{Minisum-outliers} problem thereby proving the result. Let $(\mathcal{G}=(U,E), k)$ be an arbitrary instance of the $d$-\textsc{Clique} problem. Let $U = \{u_1, \cdots, u_n\}$ and $E = \{e_1, \cdots, e_m\}$. We define the corresponding instance $(\mathcal{V}, \mathcal{C}, n^*, m^*, t)$ of the \textsc{Minisum-outliers} problem as follows: 
 \[ \mathcal{V} = \{\mathcal{S}_1, \cdots, \mathcal{S}_m\}, \mathcal{C} = \{c_1, \cdots, c_n\}, \mathcal{S}_i = \{ c_j : u_j \notin e_i \}~\forall i\in [n],  \]
 \[ n^* = {k \choose 2}, \overline{m} = k, t = (k-2){k \choose 2}\] 
 We claim that the two instances are equivalent. In the forward direction, suppose $W\subset U$ forms a clique with $|W|=k$, and let $Q$ denote the set of edges that have both endpoints in $W$. Then the committee $\mathcal{C}^* = \{ c_i : u_i \notin W\}$ along with the set of non-outliers $\mathcal{V}^* = \{ \mathcal{S}_i : u_i \in Q \}$ achieves minisum score of $(k-2){k \choose 2}$.
 
 In the reverse direction, suppose there exist a set of non-outliers $\mathcal{V}^* \subset \mathcal{V}$ and a committee $\mathcal{C}^* \subset \mathcal{C}$ such that $h(\mathcal{V}^*, \mathcal{C}^*) \le (k-2){k \choose 2}$. We claim that the vertices in $W = \{u_i : c_i\in \overline{\mathcal{C}}\}$ form a clique with edge set $Q=\{e_i: \mathcal{S}_i\in\mathcal{V}^*\}$. If not then, there exist a vote $\mathcal{S}_j\in \mathcal{V}^*$ such that $h(\mathcal{S}_j, \mathcal{C})\ge k$. On the other hand we have $h(\mathcal{S}_i,\mathcal{C})\ge k-2$ for every $\mathcal{S}_i\in \mathcal{V}^*$. This implies that $h(\mathcal{V}^*, \mathcal{C}^*) > (k-2){k \choose 2}$, which is a contradiction.
 
We define $t=0$ for the \disapprov{} and \netdisapprov{} voting rules. It is easily checked that the details are analogous.
\end{proof}

We now turn to the combination $(m^*+ n^* + t)$. Again, we have \WO{}-hardness for all voting rules considered. 

\begin{restatable}{lemma}{Whardtmstarnstar}\label{lem:hard_tkno}
 Let $s \in \{$Minisum, \Netdisapprov{}, \Disapprov{}$\}$. Then the $s$-\textsc{Outliers} problem is \WOH{}, when parameterized by $(t+m^*+n^*)$, even when every vote approves exactly two candidates.
\end{restatable}

We now consider the other combinations involving $t$ that are non-trivial for voting rules different from minisum, namely $(t + \overline{n} + m^*)$ and $(t + \overline{n} + \overline{m})$. Lemma\nobreakspace \ref {lem:hard_nbar_mstar} below shows \WO{}-hardness for the first combination, for the \Disapprov{} voting rule. This also shows \WO{}-hardness with respect to the parameter $(\overline{n} + m^*)$ alone, for the minisum voting rule. 

%NewResult
\begin{lemma}\label{lem:hard_nbar_mstar}
The \textsc{Minisum-Outliers} problem is \WOH{} when parameterized by $(\overline{n} + m^*)$.  Also, the \textsc{\Disapprov{}-Outliers} problem is \WOH{} parameterized by $(m^*+\overline{n})$, even when the target score $t=0$. In particular, $s$-\textsc{Outliers} is  \WOH{} parameterized by $(m^*+\overline{n}+t)$.
\end{lemma}

\begin{proof}
To begin with, let us prove the result for the \textsc{Minisum-outliers} problem. We exhibit a parameterized reduction from the $d$-\textsc{Clique} problem to the \textsc{Minisum-outliers} problem thereby proving the result. Let $(\mathcal{G}=(U,E), k)$ be an arbitrary instance of the $d$-\textsc{Clique} problem. Let $U = \{u_1, \cdots, u_n\}$ and $E = \{e_1, \cdots, e_m\}$. We define the corresponding instance $(\mathcal{V}, \mathcal{C}, n^*, m^*, t)$ of the \textsc{Minisum-outliers} problem as follows: 
 \[ \mathcal{V} = \{\mathcal{S}_1, \cdots, \mathcal{S}_n\}, \mathcal{C} = \{c_1, \cdots, c_m\}, \mathcal{S}_i = \{ c_j : u_i \notin e_j \}~\forall i\in [n],\]
 \[ \overline{n} = k, m^* = {k \choose 2}, t = (m-{k\choose 2})(n-k)\] 
 We claim that the two instances are equivalent. In the forward direction, suppose $W\subset U$ forms a clique with $|W|=k$, and let $Q$ denote the set of edges that have both endpoints in $W$. Then the committee $\mathcal{C}^* = \{ c_i : e_i \in Q\}$ and the set of non-outliers $\mathcal{V}^* = \{ \mathcal{S}_i : u_i \notin W \}$ achieves minisum score of $(m-{k\choose 2})(n-k)$.
 
 In the reverse direction, suppose there exist a set of outliers $\overline{\mathcal{V}} \subset \mathcal{V}$ and a committee $\mathcal{C}^* \subset \mathcal{C}$ such that $h(\mathcal{V}^*, \mathcal{C}^*) \le (m-{k\choose 2})(n-k)$. We claim that the vertices in $W = \{u_i : \mathcal{S}_i\in \overline{\mathcal{V}}\}$ form a clique. If not then, there exist a candidate $x\in\mathcal{C}^*$ that is not approved by at least one vote in $\mathcal{V}^*$. However, this implies that $h(\mathcal{V}^*, \mathcal{C}^*) > (m-{k\choose 2})(n-k)$, which is a contradiction.
 
We define $t=0$ for the \disapprov{} voting rule. It is easily checked that the details are analogous.
\end{proof}

Next, we show that \textsc{\Netdisapprov{}-outliers} is \WOH{} with respect to the combined parameter $(\overline{n}+t+m^*)$.

\begin{lemma}\label{thm:wh_tko_net}
 The \textsc{\Netdisapprov{}-outliers} problem is \WOH{}, when parameterized by $(\overline{n}+m^*)$, even when the target score is $0$. In particular, \textsc{\Netdisapprov{}-outliers} problem is \WOH{}, when parameterized by $(t+\overline{n}+m^*)$.
\end{lemma}

\begin{proof}
We reduce the \textsc{Clique} problem to the \textsc{Net-approval-outlier} problem. Let $(\mathcal{G}=(U,E), k)$ be an arbitrary instance of the \textsc{Clique} problem. Let $U = \{u_1, \cdots, u_n\}$. We define the corresponding instance $(\mathcal{V}, \mathcal{C}, n^*, m^*, t)$ of the \textsc{Net-approval-outlier} problem as follows: 
 \[ \mathcal{C} = \{c_e : e\in E\} \cup D, \text{ where } |D| = {k \choose 2},\]
 \[\mathcal{V} = \{\mathcal{S}_1, \cdots, \mathcal{S}_n\} \cup \{ \mathcal{T}_1, \cdots, \mathcal{T}_{n+2k} \}, \]
 \[ \mathcal{S}_i = \{c_e : u_i\notin e\}, \forall i\in[n], ~\mathcal{T}_j = D, \forall j\in\left[{k \choose 2}\right],\]
 \[\overline{n} = k, ~m^* = 2{k \choose 2}, ~t=0 \] 
 We claim that the two instances equivalent. In the forward direction, suppose $\mathcal{G}$ has a clique on a subset of vertices $W\subset U$ of size $k$; let $Q$ be the clique edges. Then we define the set of outlier votes to be $\overline{\mathcal{V}} = \{\mathcal{S}_i : u_i \in W\}$ and the committee to be the set of candidates $\mathcal{C}^* = D \cup \{c_e : e \in Q\}$. Now we have $\sum_{\mathcal{S} \in \mathcal{V}^*}|\mathcal{S} \cap \mathcal{C}^*| = (n-k){k \choose 2} + (n+2k)|D| = (n+k){k \choose 2}$ and $\sum_{\mathcal{S} \in \mathcal{V}^*}|\mathcal{S}^c \cap \mathcal{C}^*| = (n-k)|D| + (n+2k){k \choose 2} = (n+k){k \choose 2}$ thereby achieving the net approval score $t$ of $0$. 
 
 In the reverse direction, suppose we have a set of outlier votes $\overline{\mathcal{V}}$ and a committee $\mathcal{C}^*$ which achieves a net approval score $t$ of $0$. First observe that we can assume without loss of generality that $D$ is a subset of $\mathcal{C}^*$ since irrespective of the outliers chosen, every candidate in $D$ receives at least $(n+k)$ approvals and every candidate not in $D$ receives at most $n$ approvals. Now since the committee $\mathcal{C}^*$ contains $D$, for every $j\in[n+2k]$, the vote $\mathcal{T}_j$ contributes at most $|D|-{k \choose 2} = 0$ to the net approval score, whereas, for every $i\in[n]$, the vote $\mathcal{S}_i$ contributes at least ${k \choose 2} - |D| = 0$ to the net approval score. Hence, we may assume without loss of generality that $\mathcal{T}_j$ belongs to the set of non-outliers $\mathcal{V}^*$ for every $j\in[n+2k]$. Now we claim that the set of edges $Q = \{ e : c_e \in \mathcal{C}^*\}$ must form a clique on the set of vertices $W = \{u : \mathcal{S}_u \in \overline{\mathcal{V}}\}$. If not then, $\sum_{\mathcal{S} \in \mathcal{V}^*}|\mathcal{S} \cap \mathcal{C}^*| < (n+k){k \choose 2}$ and $\sum_{\mathcal{S} \in \mathcal{V}^*}|\mathcal{S}^c \cap \mathcal{C}^*| > (n+k){k \choose 2}$ thereby making the net approval score $t$ strictly more than $0$.
\end{proof}

Finally, we show that for all rules considered, the $s$-\textsc{Outliers} problem is \WOH{} when parameterized by $(\overline{n}+\overline{m})$.  In particular, for the \textsc{\Netdisapprov{}-outliers} problem, we have \WO{}-hardness even when $t=0$ and every candidate is approved in exactly two votes. 

\begin{lemma}\label{lem:hard_o}
Let $s \in \{$Minisum, \Disapprov{}, \Netdisapprov{}$\}$. Then $s$-\textsc{Outliers} is \WOH{}, when parameterized by $(\overline{n}+\overline{m})$, even when every candidate is approved in exactly two votes. Further, the \textsc{\Netdisapprov{}-outliers} problem, parameterized by $(\overline{n}+\overline{m})$, is \WOH even when $t=0$ and every candidate is approved in exactly two votes. In particular, the \textsc{\Netdisapprov{}-outliers} problem is \WOH{} parameterized by $(\overline{n}+\overline{m}+t)$.
\end{lemma}

\begin{proof}
First let us prove the result for the \textsc{Minisum-outliers} problem. We exhibit a parameterized reduction from the $d$-\textsc{Clique} problem to the \textsc{Minisum-outliers} problem thereby proving the result. Let $(\mathcal{G}=(U,E), k)$ be an arbitrary instance of the $d$-\textsc{Clique} problem. Let $U = \{u_1, \cdots, u_n\}$ and $E = \{e_1, \cdots, e_m\}$. We define the corresponding instance $(\mathcal{V}, \mathcal{C}, n^*, m^*, t)$ of the \textsc{Minisum-outliers} problem as follows: 
 \[ \mathcal{V} = \{\mathcal{S}_1, \cdots, \mathcal{S}_n\}, \mathcal{C} = \{c_1, \cdots, c_m\}, \mathcal{S}_i = \{ c_j : u_i \in e_j \}, \forall i\in [n],  \]
 \[ \overline{n} = k, \overline{m} = {k \choose 2}, t = (n-k)\left(m-{k \choose 2}-d\right)\] 
 We claim that the two instances are equivalent. In the forward direction, suppose $W\subset U$ forms a clique with $|W|=k$, and let $Q$ denotes the set of edges that have both endpoints in $W$. Consider the committee $\mathcal{C}^* = \{ c_i : e_i \notin Q\}$ and the set of outliers $\overline{\mathcal{V}} = \{ \mathcal{S}_i : u_i \in W \}$.  Consider now a vote $\mathcal{S}_j \in \mathcal{V}^*$. Note that $|\mathcal{C}^* \cap \mathcal{S}_j| = d$, since every edge incident on $u_j$ is an edge that does not belong to $Q$, by the definition of $\mathcal{C}^*$. Further, this also implies that $|\mathcal{C}^* \setminus \mathcal{S}_j| = \left(m-{k \choose 2}-d\right)$. Therefore, $h(\mathcal{S}_j, \mathcal{C}^*) = \left(m-{k \choose 2}-d\right).$ Hence we have $h(\mathcal{V}^*, \mathcal{C}^*) = (n-k)\left(m-{k \choose 2}-d\right).$
 
For the reverse direction, suppose there exist a set of non-outliers $\mathcal{V}^* \subset \mathcal{V}$ and a committee $\mathcal{C}^* \subset \mathcal{C}$ such that $h(\mathcal{V}^*, \mathcal{C}^*) \le (n-k)\left(m-{k \choose 2}-d\right)$. We claim that the vertices in $W = \{u_i : \mathcal{S}_i\in \overline{\mathcal{V}}\}$ form a clique. If not then there exists a vote $\mathcal{S}\in\mathcal{V}^*$, such that $h(\mathcal{S}, \mathcal{C}^*) > \left(m-{k \choose 2}-d\right)$. However, for every vote $\mathcal{S}^\prime\in\mathcal{V}^*$, we have $h(\mathcal{S}^\prime, \mathcal{C}^*) \ge \left(m-{k \choose 2}-d\right)$. This makes $h(\mathcal{V}^*, \mathcal{C}^*) > (n-k)\left(m-{k \choose 2}-d\right)$ which is a contradiction.

The proofs for the other rules are identical, except for the values of the target score. We define $t=(n-k)\left(m-{k \choose 2}-d\right)$ for the \disapprov{} voting rule. For the \netdisapprov{} voting rule, we add $(m-2d)$ many dummy candidates who are approved by every vote. We keep $\overline{m} = {k \choose 2}$ and make $t=0$.  It is easily checked that the remaining details are analogous.
\end{proof}

% \Cref{thm:hard_o} immediately gives us the following corollary.
% 
% \begin{corollary}\label{cor:2_vot}
%  The \textsc{Minisum-outliers} problem is in \NPC even when every candidate is approved by exactly two voters.
% \end{corollary}
%
%Lemma\nobreakspace \ref {lem:hard_tkno} immediately gives us the following corollary.
%
%\begin{corollary}
% There does not exist any EPTAS for the \textsc{Minisum-outliers} problem.
%\end{corollary}

\subsection{Proofs of the Main Theorems}

We are now ready to present the proofs of~Theorems\nobreakspace  \ref {thm:main1} to\nobreakspace  \ref {thm:main3} . First, we recall the dichotomy for the minisum voting rule.

\DichotomyTheorem*

\begin{proof}
Since $m^* + \overline{m} = m$ and $m^* + \overline{n} = n$, the tractability results follow from~Proposition\nobreakspace \ref {prop:fpt_mn} and Lemma\nobreakspace \ref {lem:fpt_t}. Now, we only have to consider subsets of parameters $\QQ$ such that:
\begin{enumerate}[topsep=0pt,itemsep=0pt]
	\item $\QQ$ does not contain both $m^*$ and $\overline{m}$
	\item $\QQ$ does not contain both $n^*$ and $\overline{n}$
	\item $\QQ$ does not contain both $t$ and $\overline{n}$
\end{enumerate}
Among such choices of $\QQ$, we have the following cases.
\begin{enumerate}[topsep=0pt,itemsep=0pt]
	\item Suppose $t \in \QQ$. Then $\QQ$ is either a subset of $\QQ_1 = \{t,n^*,\overline{m}\}$ or a subset of $\QQ_2 = \{t,n^*,m^*\}$. The hardness for all of these cases follow from~Lemma\nobreakspace \ref {lem:hard_mbar_nstar} and~Lemma\nobreakspace \ref {lem:hard_tkno}, respectively. 
	\item Suppose $\overline{n} \in \QQ$. Then $\QQ$ is either a subset of $\QQ_1 = \{\overline{n},\overline{m}\}$ or a subset of $\QQ_2 = \{\overline{n},m^*\}$. The hardness for all of these cases follow from~Lemma\nobreakspace \ref {lem:hard_nbar_mstar} and~Lemma\nobreakspace \ref {lem:hard_o}, respectively.
	\item If neither $t$ nor $\overline{n}$ belongs to $\QQ$, then $\QQ$ is either a subset of $\QQ_1 = \{n^*,m^*\}$, or $\QQ_2 = \{n^*,\overline{m}\}$. Note that these cases are already subsumed by Case (1) above.
\end{enumerate}
This completes the proof of the theorem. 
\end{proof}

Now, we turn to the case of the \netdisapprov{} voting rule.

\DichotomyTheoremNetApproval*

\begin{proof}
Since $m^* + \overline{m} = m$ and $m^* + \overline{n} = n$, the tractability results follow from~Proposition\nobreakspace \ref {prop:fpt_mn}. Now, we only have to consider subsets of parameters $\QQ$ such that $\QQ$ does not contain both $m^*$ and $\overline{m}$, and $\QQ$ does not contain both $n^*$ and $\overline{n}$. Among such choices of $\QQ$, we have the following cases.
\begin{enumerate}[topsep=0pt,itemsep=0pt]
	\item Suppose $n^* \in \QQ$. Then $\QQ$ is either a subset of $\QQ_1 = \{n^*,\overline{m}, t\}$ or a subset of $\QQ_2 = \{n^*,m^*,t\}$. The hardness for all of these cases follow from~Lemma\nobreakspace \ref {lem:hard_mbar_nstar} and~Lemma\nobreakspace \ref {lem:hard_tkno}, respectively. 
	\item Suppose $\overline{n} \in \QQ$. Then $\QQ$ is either a subset of $\QQ_1 = \{\overline{n},\overline{m},t\}$ or a subset of $\QQ_2 = \{\overline{n},m^*,t\}$. The hardness for all these cases follow from~Lemma\nobreakspace \ref {lem:hard_o} and~Lemma\nobreakspace \ref {thm:wh_tko_net}, respectively. 
	\item If neither $n^*$ nor $\overline{n}$ belongs to $\QQ$, then $\QQ$ is either a subset of $\QQ_1 = \{t,m^*\}$, or $\QQ_2 = \{t,\overline{m}\}$. Note that these cases are already subsumed by the cases above.
\end{enumerate}
This completes the proof of the theorem. 
\end{proof}

Finally, we turn to the case of the \disapprov{} voting rules.

\NearDichotomyTheorem*

\begin{proof}
Since $m^* + \overline{m} = m$ and $m^* + \overline{n} = n$, the tractability results follow from~Proposition\nobreakspace \ref {prop:fpt_mn}. Now, we only have to consider subsets of parameters $\QQ$ such that $\QQ$ does not contain both $m^*$ and $\overline{m}$, and $\QQ$ does not contain both $n^*$ and $\overline{n}$. Among such choices of $\QQ$, we have the following cases.
\begin{enumerate}[topsep=0pt,itemsep=0pt]
	\item Suppose $n^* \in \QQ$. Then $\QQ$ is either a subset of $\QQ_1 = \{n^*,\overline{m}, t\}$ or a subset of $\QQ_2 = \{n^*,m^*,t\}$. The hardness for all of these cases follows from~Lemma\nobreakspace \ref {lem:hard_mbar_nstar} and~Lemma\nobreakspace \ref {lem:hard_tkno}, respectively. 
	\item Suppose $\overline{n} \in \QQ$. Then $\QQ$ is either a subset of $\QQ_1 = \{\overline{n},m^*,t\}$ or $\QQ_2 = \{\overline{n},\overline{m},t\}$. The hardness for all subsets of $\QQ_1$ follows from~Lemma\nobreakspace \ref {lem:hard_nbar_mstar}. The status for $\QQ_2$, as stated in the theorem, is open. The hardness for all strict subsets of $\QQ_2$ follows from~Lemma\nobreakspace \ref {lem:hard_o} and the cases that are already resolved. 
	\item If neither $n^*$ nor $\overline{n}$ belongs to $\QQ$, then $\QQ$ is either a subset of $\QQ_1 = \{t,m^*\}$, or $\QQ_2 = \{t,\overline{m}\}$. Note that these cases are already subsumed by the cases above.
\end{enumerate}
This completes the proof of the theorem. 
\end{proof}

\section{Approximation Results}

In this section, we describe our results in the context of approximation algorithms, where our goal is to minimize the target score, given a committee size $m^*$ and a budget $\overline{n}$ for the number of outliers as a part of input. In the first part, we focus on the minisum voting rule, and show an $\varepsilon \overline{m}$ approximation algorithm, for every constant $\eps>0$. We also have a $(1 + \eps)$-approximation algorithm whose running time is $n^{O\left(\nfrac{\log m}{\eps^2}\right)}$, for every constant $\eps>0$. Subsequently, we show that for all the voting rules, the target score is \NPH{} to approximate within any factor, as it is \NPH{} already to determine if $t = 0$. 

\subsection{Approximation Algorithms}

We now turn to some approximation algorithms for the \textsc{Minisum-Outliers} problem. Our first result is the following.

\begin{restatable}{theorem}{MinisumApprox}\label{thm:Minisum_approx}
 There is a $\eps \overline{m}$-approximation algorithm for the \textsc{Minisum-outliers} problem, for every constant $\eps>0$.
\end{restatable}

Our next result is an approximation scheme with a subexponential running time, loosely based on the framework introduced in~\cite{LKLB14}. To this end, we will need the following lemma. 

\begin{lemma}\label{lem:prob}
 Let $\eps$ be any positive constant. Then, given a set of votes $\mathcal{V}$, there exists a subset $\mathcal{V}^\prime \subset \mathcal{V}$ of size $O\left(\frac{\log m}{\eps^2}\right)$ such that $h(\mathcal{V}, c(\mathcal{V}^\prime)) \le (1+\eps) h(\mathcal{V}, c(\mathcal{V}))$, where $c(\mathcal{V})$ and $c(\mathcal{V}^\prime)$ are the committees chosen by Minisum voting rule on $\mathcal{V}$ and $\mathcal{V}^\prime$ respectively.
\end{lemma}

\begin{proof}
 We prove the existence of a subset $\mathcal{V}^\prime \subset \mathcal{V}$ of size $O(\nfrac{\log m}{\eps^2})$ such that $h(\mathcal{V}, c(\mathcal{V}^\prime)) \le (1+20\eps) h(\mathcal{V}, c(\mathcal{V}))$. Applying the weaker inequality with $\eps^\prime = \frac{\eps}{20}$ then proves the Lemma. We assume without loss of generality that $\eps \le \frac{1}{40}.$ We prove the statement using probabilistic methods. We pick $r$ votes from $\mathcal{V}$ uniformly at random to form the set $\mathcal{V}^\prime$. Let $c \in \mathcal{C}$ be any arbitrary candidate. Let $\mathbbm{1}_{\mathcal{V}}(c)$ and $\mathbbm{1}_{\mathcal{V}^\prime}(c)$ be the number of approvals that the candidate $c$ receives from the votes in $\mathcal{V}$ and $\mathcal{V}^\prime$ respectively. There exist an $r=O(\nfrac{\log m}{\eps^2})$ such that the following holds (Theorem 5 in~\cite{DeyB15}). 
 \[ \Pr \left[ \forall c\in \mathcal{C}, \left|\frac{n}{r}\mathbbm{1}_{\mathcal{V}^\prime}(c) - \mathbbm{1}_{\mathcal{V}}(c)\right| \le \eps n \right] > 0 \] 
 Now assume that for every candidate $c\in \mathcal{C}$, we have $\left|\frac{n}{r}\mathbbm{1}_{\mathcal{V}^\prime}(c) - \mathbbm{1}_{\mathcal{V}}(c)\right| \le \eps n$ and the statement above shows that there exists such a set $\mathcal{V}^\prime$. With this assumption, we now prove the result. Suppose $c(\mathcal{V}^\prime) \setminus c(\mathcal{V}) = \{z_i: i\in[\ell]\}$. Then there exists exactly $\ell$ many candidates $z_i^\prime, i\in [\ell]$  such that, $c(\mathcal{V}) \setminus c(\mathcal{V}^\prime) = \{z_i^\prime: i\in [\ell]\}$ since both $c(\mathcal{V})$ and $c(\mathcal{V}^\prime)$ are committees of same size $m^*$. The assumption that for every candidate $c\in \mathcal{C}$, we have $\left|\frac{n}{r}\mathbbm{1}_{\mathcal{V}^\prime}(c) - \mathbbm{1}_{\mathcal{V}}(c)\right| \le \eps n$ gives us the following: 
 \[ h(\mathcal{V}, c(\mathcal{V}^\prime)) \le h(\mathcal{V}, c(\mathcal{V})) + 4\eps\ell n \] 
 Now we consider the following two cases.
 \begin{itemize}
  \item Case 1: $\mathbbm{1}(z_i^\prime) \le \frac{n}{4}$ for every $i\in[\ell]:$
 
 We have $h(\mathcal{V}, c(\mathcal{V})) \ge \frac{3n}{4}\ell$ and thus $h(\mathcal{V}, c(\mathcal{V}^\prime)) \le (1+\frac{16}{9}\eps) h(\mathcal{V}, c(\mathcal{V})).$
  \item Case 2: $\mathbbm{1}(z_i^\prime) \ge \frac{n}{4}$ for some $i\in[\ell]:$
 
 For every $j\in[\ell]$, the candidate $z_j$ is in the committee $c(\mathcal{V}^\prime)$ implies the following.
\begin{eqnarray*}
  \mathbbm{1}(z_i) &\ge& \frac{n}{4} - 2\eps n \text{ for every } i\in[\ell]\\
  \Rightarrow h(\mathcal{V}, c(\mathcal{V})) &\ge& (\frac{n}{4} - 2\eps n)\ell\\
  &\ge& \frac{n}{5}\ell\\
  \Rightarrow h(\mathcal{V}, c(\mathcal{V}^\prime)) &\le& (1+20\eps) h(\mathcal{V}, c(\mathcal{V}))
 \end{eqnarray*}
 The third line follows from the the assumption that $\eps \le \frac{1}{40}.$
 \end{itemize}
\end{proof}

We now turn to our approximation algorithm. 

\begin{theorem}\label{thm:Minisum_as}
 There is a $(1+\eps)$-approximation algorithm for the \textsc{Minisum-outliers} problem running in time $n^{O\left(\nfrac{\log m}{\eps^2}\right)}$, for every constant $\eps>0$.
\end{theorem}

\begin{proof}
 Let $\mathcal{V}^* \subset \mathcal{V}$ be the set of non-outliers in the optimal solution. We apply Lemma\nobreakspace \ref {lem:prob} to the set $\mathcal{V}^*$. The algorithm guesses the set $\mathcal{V}^\prime$ by trying all possible $n^{O(\nfrac{\log m}{\eps^2})}$ subsets of $\mathcal{V}$ (since we do not know $\mathcal{V}^*$). Let $c(\mathcal{V}^\prime)$ be the Minisum committee of the votes in $\mathcal{V}^\prime.$ The algorithm returns the set $\mathcal{V}^{\prime\prime}$ of the $n^*$ votes that have the smallest Hamming distances from $c(\mathcal{V}^\prime)$ and the committee $c(\mathcal{V}^{\prime\prime})$. Then we have the following.
 \[ h(c(\mathcal{V}^{\prime\prime}), \mathcal{V}^{\prime\prime}) \le h(c(\mathcal{V}^\prime), \mathcal{V}^{\prime\prime}) \le h(c(\mathcal{V}^\prime), \mathcal{V}^*) \le (1+\eps)h(c(\mathcal{V}^*), \mathcal{V}^*) \]
 The last inequality follows from Lemma\nobreakspace \ref {lem:prob}.
\end{proof}

\subsection{Hardness of Approximation}

In the previous section, we demonstrated some approaches for approximating the target score in the context of minisum voting rule. We now show that the other voting rules are inapproximable in a very strong sense. Our first hardness result is that it is \NPH{} to check even if there is a subset of outliers and a corresponding committee for which the \netdisapprov{} score is zero. This follows from Lemma~\ref{lem:hard_o}.

\begin{corollary}
\label{thm:no_apx_net}
 There does not exist any polynomial time $\alpha(m,n)$-approximation algorithm for the \textsc{\Netdisapprov{}-outliers} problem for any computable function $\alpha(\cdot,\cdot)$, unless $\Pshort=\NPshort$.
\end{corollary}

Next, we show that it is \NPH{} to check even if there is a subset of outliers and a corresponding committee for which the \disapprov{} score is zero. For this, we reduce from a problem closely related to \textsc{Max-Clique}, namely that of finding a largest-sized biclique in a graph $G$, which is known to be \NPC{}~\cite{garey1979computers,Johnson}.

\defproblem{Biclique}
{An undirected bipartite graph $\mathcal{G} = (V_L, V_R, E)$ and an integer $k$}
{Do there exist two subsets of vertices $A\subset V_L$ and $B\subset V_R$ with $|A|=|B|=k$ such that $(A,B)$ forms a biclique?}

\begin{lemma}\label{thm:no_apx_other}
 There does not exist any polynomial time $\alpha(m,n)$-approximation algorithm for \textsc{\Disapprov{}-outliers} for any computable function $\alpha(\cdot,\cdot)$, unless $\Pshort=\NPshort$.
\end{lemma}

\begin{proof}
We will show that deciding whether the optimal disapproval score is zero or not, is \NPC, thereby implying the result. The problem is clearly in \NP. To prove \NP-hardness, we reduce the \textsc{Biclique} problem to the \textsc{\Disapprov{}-outliers} problem. Let $(\mathcal{G}=(V_L, V_R, E), k)$ be an arbitrary instance of the \textsc{Biclique} problem. Let $V_L = \{u_1, \cdots, u_n\}$ and $V_R = \{u_1^\prime, \cdots, u_m^\prime\}$. We define the corresponding instance $(\mathcal{V}, \mathcal{C}, n^*, m^*, t)$ of the \textsc{\Disapprov{}-outliers} problem as follows: 
 \[ \mathcal{V} = \{\mathcal{S}_1, \cdots, \mathcal{S}_n\}, \mathcal{C} = \{c_1, \cdots, c_m\}, \]
 \[ \mathcal{S}_i = \{ c_j : (u_i, u_j^\prime) \in E \}~\forall i\in [n],n^* = k, m^* = k, t = 0 \] 
 We claim that the two instances are equivalent. In the forward direction, suppose $A\subset V_L$ and $B\subset V_R$ with $|A|=|B|=k$ forms a biclique. Consider the set of non-outliers to be $\mathcal{V}^* = \{ \mathcal{S}_i : u_i \in A \}$ and the committee $\mathcal{C}^* = \{ c_j : u_j^\prime \in B \}$. This achieves \disapprov{} score of $0$.
 
 In the reverse direction, suppose there exist a set of non-outliers $\mathcal{V}^*$ and a committee $\mathcal{C}^*$ which achieves disapproval score of $0$. Without loss of generality, we assume that $|\mathcal{V}^*|=k$ since removing votes from the set of non-outliers can only reduce the score. Consider $A = \{u_i : \mathcal{S}_i \in \mathcal{V}^*\} \subset V_L$ and $B = \{ c_j : c_j \in \mathcal{C}^* \}\subset V_R$. We claim that $A$ and $B$ must form a biclique. Indeed, otherwise the \disapprov{} score would be strictly more than $0$.
\end{proof}

\section{Concluding Remarks and Future Directions}
\label{sec:conclusions}

We have studied the problem of determining the a committee that is good for a large portion of the electorate, based on some common voting rules. We begin this is section by showing that our choice of the notion of disapproval was an arbitrary one, by demonstrating the following lemma. We use the notation \disapprov{}$^\prime$ to denote the voting rule given by the corresponding bottom row in Table\nobreakspace \ref {table:approvalscores}.

\begin{restatable}{lemma}{Equivalence}\label{lem:translate-disapprov}
 There is a polynomial time approximation preserving reduction from \textsc{\Disapprov{}-outliers} problem to \textsc{\Disapprov{}$^\prime$-outliers} problem and vice-versa. Moreover, the reductions are such that $m^*$ in the original instance becomes $\overline{m}$ in the reduced instance, $\overline{m}$ in the original instance becomes $m^*$ in the reduced instance, and every other parameter remain exactly same. 
\end{restatable}

\begin{proof}
 We provide the reduction from \textsc{\Disapprov{}-outliers} to \textsc{\Disapprov{}$^\prime$-outliers}. The other cases are exactly same. Let $(\mathcal{V},\mathcal{C}, n^*, m^*, t)$ be an arbitrary instance of \textsc{\Disapprov{}-outliers}. We define the corresponding instance $(\mathcal{V}^\prime,\mathcal{C}^\prime, n^{*\prime}, m^{*\prime}, t^\prime)$ of \textsc{\Disapprov{}$^\prime$-outliers} as follows. 
 \[ \mathcal{C}^\prime = \mathcal{C}, \mathcal{S}_i^\prime = \{ x\in\mathcal{C}^\prime : x\notin \mathcal{S}_i \}, n^{*\prime} = n^*, m^{*\prime} = \overline{m}, t^\prime = t \] 
 We claim that the two instances are equivalent. If there exist a $\mathcal{V}^*\subset\mathcal{V}$ and $\mathcal{C}^*\subset\mathcal{C}$ such that $s_{disapproval}(\mathcal{V}^*, \mathcal{C}^*) \le t$ then $s_{disapproval^\prime}(\mathcal{V}^*, \overline{\mathcal{C}})\le t$ and vice versa.
\end{proof}

There are some very special cases that are solvable in polynomial time. In particular, when every candidate is approved by at most one vote, or when every vote approves at most one candidate, then the \textsc{Minisum-outliers} problem can be solved in polynomial amount of time.

\begin{lemma}
 The \textsc{Minisum-outliers} problem is in \Pshort when every vote approves at most one candidate.
\end{lemma}

\begin{proof}
 The following greedy algorithm computes the committee and the set of outliers which achieves lowest minisum score. First choose the $m^*$ candidates with maximum number of approvals. Let $\mathcal{C}^*$ be the chosen set of $m^*$ candidates. Let $U$ be the set of votes who approve some candidate in $\mathcal{C}^*$, $W$ be the set of votes who approve some candidate in $\overline{\mathcal{C}}(=\mathcal{C}\setminus \mathcal{C}^*)$, and $Y = \mathcal{V} \setminus (U \cup W)$. Then, if $|W| \ge \overline{n},$ then define the set of outliers to be $\overline{\mathcal{V}}$ for any arbitrary $\overline{\mathcal{V}} \subseteq W$ with $|\overline{\mathcal{V}}|=\overline{n}$. If $|W|<\overline{n}$ and $|W| + |Y| \ge \overline{n},$ then define the set of outliers to be $\overline{\mathcal{V}} = W \cup Z,$ for any arbitrary $Z \subseteq Y$ with $|Z| = \overline{n} - |W|$. Otherwise, define the set of outliers to be $W \cup Y \cup A,$ for any arbitrary $A \subseteq U$ with $|A|=\overline{n}-|W|+|Y|.$
\end{proof}

\begin{restatable}{lemma}{Polycases}
 The \textsc{Minisum-outliers} problem is in \Pshort when every candidate is approved in at most one vote
\end{restatable}

\begin{proof} 
Consider the following dynamic programming based algorithm. Let the votes be $\mathcal{S}_1, \mathcal{S}_2, \cdots, \mathcal{S}_n.$ We define $M(i, j, \ell)$ to be a set of $j$ outliers (call it an optimal set of outliers) that minimizes the sum of Hamming distances when we are selecting a committee of size $\ell$ considering only votes in $\{\mathcal{S}_1, \mathcal{S}_2, \cdots, \mathcal{S}_i\}.$ We compute the value of $M(i, j, \ell)$. Consider the case when the vote $\mathcal{S}_i$ belongs to an optimal set of outliers when we are selecting a committee of size $\ell$ considering only votes in $\{\mathcal{S}_1, \mathcal{S}_2, \cdots, \mathcal{S}_i\}.$ In this case, $M(i, j, \ell) = M(i-1, j-1, \ell).$ Consider the other case when the vote $\mathcal{S}_i$ does not belong to any optimal set of outliers when we are selecting a committee of size $\ell$ considering only votes in $\{\mathcal{S}_1, \mathcal{S}_2, \cdots, \mathcal{S}_i\}.$ Then $M(i, j, \ell) = M(i-1, j, \ell - |\mathcal{S}_i|)$ if $\ell \ge |\mathcal{S}_i|$ and $M(i, j, \ell) = M(i-1, j, 0)$ otherwise, since when the vote $\mathcal{S}_i$ is not an outlier then we can assume without loss of generality that the committee has maximum overlap with the vote $\mathcal{S}_i.$ Base cases are taken care from the fact that $M(i, j, \ell)$ can be computed in polynomial time irrespective of the values of $\ell$ and $j$ when $i$ is a constant. We compute the value of $M(i, j, \ell)$ for every $i\in[n], j \le i, \ell\in[m].$  $M(n, \overline{n}, m^*)$ gives us the optimal set of outliers (which also in turn gives us the committee to choose).
\end{proof}

% \begin{lemma}
%  The \textsc{Minisum-outliers} problem is in \Pshort when every candidate is approved in at most one vote, or when every voter approves at most one candidate. 
% \end{lemma}

\paragraph{Future Directions} ~~One pertinent open problem is to close the only unresolved case of parameterized complexity of $s$-\textsc{Outliers} for the  \disapprov{} voting rule when parameterized by $(\overline{m}, \overline{n}, t)$. It is also open to improve, or prove lower bounds, for our approximation results in the context of the minsum voting rule. Refining the \FPT{} fragment of our dichotomy to kerenlization is also an exciting direction for further investigation. Finally, we are also interested in extending our dichotomous results to a dichotomy based on more parameters. We have made progress towards that by incorporating two more parameters, namely the maximum number of candidates that any vote approves and maximum number of approvals that any candidate obtains. 

\newpage
\longversion{\bibliographystyle{apalike}}
\shortversion{\bibliographystyle{unsrt}}
\bibliography{multiwinner}

%Switched to biblatex (see preamble)
%\bibliographystyle{plain}
%\bibliography{paperid}

% \newpage
\longversion{
\section*{Appendix}

Now we provide the missing proofs.

\FPTmnLemma*

\begin{proof}
 Notice that if we know the committee $\mathcal{C}^*\subset \mathcal{C}$ in the optimal solution then we can solve the problem in polynomial amount of time by choosing $n^*$ many votes as non-outliers that have the smallest score with respect to the committee $\mathcal{C}^*$. We guess the committee $\mathcal{C}^*$ of the optimal solution (by iterating over $O(2^m)$ many possibilities) and solve the problem in time $O^*(2^m)$.
 
 Notice that if we know the set of non-outliers $\mathcal{V}^*\subset \mathcal{V}$ in the optimal solution then we can solve the problem in polynomial amount of time by choosing $m^*$ many candidates that receive least disapprovals from the votes in $\mathcal{V}^*$. We guess the set of non-outliers $\mathcal{V}^*$ of the optimal solution (by iterating over $O(2^n)$ many possibilities) and solve the problem in time $O^*(2^n)$.
\end{proof}

\Whardtmstarnstar*

\begin{proof}
 First let us prove the result for the \textsc{Minisum-outliers} problem. We exhibit a parameterized reduction from the \textsc{Clique} problem to the \textsc{Minisum-outliers} problem thereby proving the result. Let $(\mathcal{G}=(U,E), k)$ be an arbitrary instance of the \textsc{Clique} problem. Let $U = \{u_1, \cdots, u_n\}$ and $E = \{e_1, \cdots, e_m\}$. We define the corresponding instance $(\mathcal{V}, \mathcal{C}, n^*, m^*, t)$ of the \textsc{Minisum-outliers} problem as follows: 
 \[ \mathcal{V} = \{\mathcal{S}_1, \cdots, \mathcal{S}_m\}, \mathcal{C} = \{c_1, \cdots, c_n\}, n^* = {k \choose 2}, m^* = k,\]
 \[ t = (k-2){k \choose 2}, \mathcal{S}_i = \{ c_j : v_j \in e_i \}, \forall i\in [m] \] 
 We claim that the two instances are equivalent. In the forward direction, suppose $W\subset U$ forms a clique with $|W|=k$ in $\mathcal{G}$, and let $Q$ denote the set of edges that have both endpoints in $W$. Consider the committee $\mathcal{C}^* = \{ c_i : u_i \in W\}$ and the set of non-outliers $\mathcal{V^*} = \{ \mathcal{S}_i : e_i \in Q \}$. This achieves the minisum score of $(k-2){k \choose 2}.$
 
 In the reverse direction, suppose there exist a set of non-outliers $\mathcal{V}^* \subset \mathcal{V}$ with $|\mathcal{V}^*|\ge {k\choose 2}$ and a committee $\mathcal{C}^* \subset \mathcal{C}$ with $|\mathcal{C}^*|=k$ such that $h(\mathcal{V}^*, \mathcal{C}^*) \le (k-2){k \choose 2}$. Without loss of generality, we assume $|\mathcal{V}^*| = {k\choose 2}$ since more votes in the set of non-outliers cannot reduce the minisum score. We claim that all the edges in $Q = \{e_j : \mathcal{S}_j \in \mathcal{V}^* \}$ have both the end points in $W=\{u_i : c_i \in \mathcal{C}^*\}$ thereby the vertices in $W$ forming a clique. Indeed otherwise, suppose that there is an edge $e_j \in Q$ for which both the end points are not in $W$. Then $h(\mathcal{S}_j, \mathcal{C}^*) > (k-2)$. On the other hand we have $h(\mathcal{S}, \mathcal{C}^*) \ge (k-2)$ for every $\mathcal{S} \in \mathcal{V}^*$. Hence we have $h(\mathcal{V}^*, \mathcal{C}^*) > (k-2){k \choose 2}$ which is a contradiction. 
 
The proofs for the other rules are identical, except for the values of the target score. We define $t=(k-2){k \choose 2}$ for the \disapprov{} and $t=(k-4){k \choose 2}$ for the \netdisapprov{} voting rules. It is easily checked that the details are analogous.
\end{proof}

\MinisumApprox*

\begin{proof}
 Fix any constant $\eps>0$. Consider the case that the minisum committee, say $\mathcal{C}^*,$ incurs a sum of Hamming distances of at most $\frac{n^*}{\eps}.$ In this case, by pigeon hole principle, there exists a vote say $\mathcal{S}_i$ whose Hamming distance from the committee $\mathcal{C}^*$ is at most $\frac{1}{\eps}$. We can guess the vote $\mathcal{S}_i$ (by iterating over all the votes) and try all possible committees of size $m^*$ which is at most $\frac{1}{\eps}$ far away from $\mathcal{S}_i$; this takes polynomial amount of time since there are only $O(m^{\frac{1}{\eps}})$ many such committees. Hence, in polynomial amount of time, we can check if the sum of Hamming distances of the minisum committee and the votes is less than $\frac{n^*}{\eps}$ or not. Now assume that $\mathcal{C}^*$ incurs a sum of the Hamming distances of at least $\frac{n^*}{\eps}.$ We choose any $n^*$ many voters and return the committee $\mathcal{X}$ consisting of $m^*$ many candidates that receive most approvals. Notice that the sum of the number of disapprovals that any candidate in the committee receives and the number of approvals that any candidate not in the committee receives is at most $n^*$. Hence, the sum of Hamming distances of $\mathcal{X}$ and the $n^*$ many votes chosen is at most $m^*n^*+(m-2m^*)n^* = \overline{m}n^*.$
\end{proof}

\Polycases*

\begin{proof} 
Consider the following dynamic programming based algorithm. Let the votes be $\mathcal{S}_1, \mathcal{S}_2, \cdots, \mathcal{S}_n.$ We define $M(i, j, \ell)$ to be a set of $j$ outliers (call it an optimal set of outliers) that minimizes the sum of Hamming distances when we are selecting a committee of size $\ell$ considering only votes in $\{\mathcal{S}_1, \mathcal{S}_2, \cdots, \mathcal{S}_i\}.$ We compute the value of $M(i, j, \ell)$. Consider the case when the vote $\mathcal{S}_i$ belongs to an optimal set of outliers when we are selecting a committee of size $\ell$ considering only votes in $\{\mathcal{S}_1, \mathcal{S}_2, \cdots, \mathcal{S}_i\}.$ In this case, $M(i, j, \ell) = M(i-1, j-1, \ell).$ Consider the other case when the vote $\mathcal{S}_i$ does not belong to any optimal set of outliers when we are selecting a committee of size $\ell$ considering only votes in $\{\mathcal{S}_1, \mathcal{S}_2, \cdots, \mathcal{S}_i\}.$ Then $M(i, j, \ell) = M(i-1, j, \ell - |\mathcal{S}_i|)$ if $\ell \ge |\mathcal{S}_i|$ and $M(i, j, \ell) = M(i-1, j, 0)$ otherwise, since when the vote $\mathcal{S}_i$ is not an outlier then we can assume without loss of generality that the committee has maximum overlap with the vote $\mathcal{S}_i.$ Base cases are taken care from the fact that $M(i, j, \ell)$ can be computed in polynomial time irrespective of the values of $\ell$ and $j$ when $i$ is a constant. We compute the value of $M(i, j, \ell)$ for every $i\in[n], j \le i, \ell\in[m].$  $M(n, \overline{n}, m^*)$ gives us the optimal set of outliers (which also in turn gives us the committee to choose).
\end{proof}
}
\end{document}